\documentclass[a4paper,USenglish]{lipics-v2021}

\usepackage[T1]{fontenc}
\usepackage{graphicx} 
\usepackage{todonotes}
\usepackage{xspace}
\usepackage{graphicx}
\usepackage[all]{xy}
\usepackage{ebproof}
\usepackage{algorithm}
\usepackage{algpseudocode}

\newcommand{\prob}{\mathsf{Pr}}

\newcommand{\prop}{\mathsf{Prop}}
\newcommand{\SFP}{\mathsf{P}}

\newcommand{\CM}{\mathcal{M}}
\newcommand{\CL}{\mathcal{L}}

\newcommand{\CI}{\mathcal{I}}

\newcommand{\CV}{\mathcal{V}}

\newcommand{\CT}{\mathcal{T}}
\newcommand{\CO}{\mathcal{O}}
\newcommand{\CA}{\mathcal{A}}
\newcommand{\adv}{\mathbb{A}}
\newcommand{\rec}{\mathsf{rec}}

\nolinenumbers

\usepackage{tikz}
\usetikzlibrary{arrows,calc,patterns,positioning,shapes}
\usetikzlibrary{decorations.pathmorphing,calc,arrows.meta, positioning}
\tikzset{
modal/.style={>=stealth',shorten >=1pt,shorten <=1pt,auto,
node distance=1.5cm,semithick},
world/.style={circle,draw,minimum size=1cm,fill=gray!15},
point/.style={circle,draw,fill=black,inner sep=0.5mm},
reflexive/.style={->,in=120,out=60,loop,looseness=#1},
reflexive/.default={5},
reflexive point/.style={->,in=135,out=45,loop,looseness=#1},
reflexive point/.default={25},
}

\tikzset{
reflexive above/.style={->,loop,in=120,out=60,looseness=#1},
reflexive above/.default={7},
reflexive below/.style={->,loop,in=240,out=300,looseness=#1},
reflexive below/.default={7},
reflexive left/.style={->,loop,in=150,out=210,looseness=#1},
reflexive left/.default={7},
reflexive right/.style={->,loop,in=30,out=330,looseness=#1},
reflexive right/.default={7}
}

\title{An Epistemic Analysis of Random Coordinated Attack}

\author{Sophia Knight}{University of Minnesota Duluth, Duluth, Minnesota}{sophia.knight@gmail.com}{}{}

\author{David Lehnherr}{Independent researcher}{edavidlehnherr@gmail.com}{}{}

\author{Sergio Rajsbaum}{Instituto de Matem\'aticas, UNAM, Mexico City, Mexico and IRIF Universit\'e Paris Cit\'e, Paris, France}{sergio.rajsbaum@gmail.com}{}{}

\authorrunning{Knight, Lehnherr, Rajsbaum}

\ccsdesc[500]{Theory of computation~Modal and temporal logics}
\ccsdesc[500]{Computing methodologies~Distributed algorithms}
\ccsdesc[500]{Theory of computation~Probabilistic computation}

\keywords{Distributed computing, Random coordinated attack, Probabilistic epistemic logic, Task solvability }
\date{}

\Copyright{CC-BY 4.0}
\ArticleNo{0}

\begin{document}
\maketitle
\begin{abstract}
The coordinated attack problem models  the challenges  of coordinating a joint action that needs to be performed in a bounded amount of time, by communicating over  unreliable links.  It is the first distributed computing problem to be proven to be unsolvable.
Analysis of coordinated attack also revealed the importance of common knowledge, a central concept of epistemic logic. However, the randomized version of coordinated attack, which is solvable, has not, to the best of our knowledge, been studied through the lens of probabilistic epistemic logic, where processes can generate randomness by flipping coins.

In this work, we present a general epistemic logic framework to study  randomized algorithms with a bounded number of rounds. It can be used to study tasks such as coordinated attack, approximate agreement, and consensus. The framework can be used to study general dynamic graph models:   synchronous systems where reliable processes execute a bounded number of rounds, and messages can be lost, as determined by an adversary.

We combine techniques from the logical characterization of dynamic networks and the notion of task solvability in distributed computing with ideas from probabilistic dynamic epistemic logic. The framework is inspired by the operational model of Varghese and Lynch that introduced randomized coordinated attack. 
More broadly, the resulting notion of probabilistic epistemic task solvability provides a foundation for the epistemic study of randomized distributed computation.

Using this framework, we analyze the algorithm of Varghese and Lynch from a knowledge-theoretic perspective, providing a more formal treatment of both the algorithm and its corresponding lower bound. As a byproduct, we  improve the lower bound to make it tight. Our lower bound proof relies on indistinguishability arguments, demonstrating that reasoning about knowledge remains important in the probabilistic setting. In addition to the tight lower bound, we give a formal semantics of the crucial notion of `information level' introduced by Varghese and Lynch by showing that it corresponds to a specific epistemic formula.

\end{abstract}

\section{Introduction}
The coordinated attack problem~\cite{Akkoyunlu75} was formalized by Jim Gray in 1978~\cite{gray78}, to identify the challenges  of coordinating an action that must be performed in a \emph{bounded} amount of time, communicating over an unreliable channel.  Such coordination is needed often in databases, networking, and other distributed systems. Thus, it was historically significant as the first distributed problem proven to be unsolvable. Moreover, it initiated the study of distributed computability and underscored the role of common knowledge, a central concept in epistemic logic~\cite{FHMV1995}. 

The model used for this classic impossibility result comprises a set of reliable processes that communicate through message exchange in synchronous rounds. It is a very weak model,  allowing failure of all links, in every round. 
The coordinated attack problem is stated in terms of three requirements, for processes starting with initial values that are binary inputs, and outputting binary decision values after some \emph{bounded} number of rounds, $R$:
\begin{description}
\item[Validity:] \hfill
\begin{enumerate}
    \item If any process starts with an initial value of 0, then 0 is the only possible decision value for every process.
    \item If all processes start with an initial value of 1, and all messages are delivered, then 1 is
the only possible decision value for every process.
\end{enumerate}
\item[Agreement:] All processes choose the same decision value. 
\item[Termination:] Each process eventually outputs a decision value.
\end{description}
Since this problem is unsolvable by a deterministic algorithm~\cite{DBLP:journals/jacm/FischerLP85}, Varghese and Lynch~\cite{VargheseL96} studied \emph{randomized coordinated attack} with slightly different correctness conditions:
\begin{description}
\item[Validity:] \hfill
\begin{enumerate}
    \item If any process starts with an initial value of 0, then 0 is the only possible decision value for all processes.
    \item If all processes start with an initial value of 1 and all messages are delivered, then 1 is
the only possible decision value for all processes.
\end{enumerate}

\item[Randomized Agreement:] The probability that all processes do not output the same decision value is smaller than $\epsilon$.
\item[Termination:] All processes eventually output a decision value.
\end{description}
Under a deterministic, non-adaptive adversary that cannot read message bits, Varghese and Lynch present an $R$-round algorithm, and show that it is essentially optimal in terms of $\epsilon$.
The algorithm has $\epsilon=1/R$ and they show that the general lower bound is $\frac{1}{R+1}$.

The impossibility of coordinated attack is frequently taught in networking, database and distributed computing courses. However,  randomized coordinated attack  deserves more attention.
 The ideas for both the algorithm and the lower bound are  elegant and simple, described subsequently by Lynch~\cite[Chapter 5]{Lynch96} and Aspnes~\cite[Chapter 8]{aspnes2025notestheorydistributedsystems}.
However, little subsequent work has focused on either the algorithm or the lower bound.  Despite their simplicity, these results merit closer examination. 
To begin with,
the (small) gap between $\frac{1}{R}$ and $\frac{1}{R+1}$ is intriguing, and
there seems to be an underlying epistemic interpretation: from Lynch's book Section 5.2.2.  (P.~88):
\begin{quote}
    ``The algorithm is based on what processes know about each other's initial
values and on what they know about each other's knowledge of the initial values,
and so on.'' 
\end{quote}
We formalize this claim and elucidate its significance by characterizing both the knowledge gained by processes in the algorithm and the knowledge that remains unattainable in the lower-bound argument.

The benefit obtained is first of all, a new exposition of the results, which formalizes the intuitive arguments, and reveals their generality.
Once a precise epistemic understanding is reached, it becomes clear  how to close the gap between the lower and upper bounds. The core contributions of the present paper are:

\begin{enumerate}
\item A probabilistic epistemic framework for randomized bounded synchronous distributed computation.
\item An epistemic characterization of information levels.
\item A tight lower bound for randomized coordinated attack.
\item  A framework that is applicable to general models of communication.
\end{enumerate}

The developed framework combines approaches from distributed computing and logical methods for reasoning about uncertainty.  We demonstrate the generality of our framework by analyzing randomized coordinated attack in the iterated immediate snapshot model of communication ($\mathcal{IIS}$).

A central contribution of this paper is a probabilistic epistemic framework for reasoning about randomized distributed computation. The framework makes it possible to analyze randomized coordinated attack through the lens of knowledge and reveals a common structure underlying both the algorithm of Varghese and Lynch and its lower bound. In this view, information levels correspond to finite levels of iterated knowledge, while the lower bound follows from fundamental constraints on the acquisition of knowledge. This unification not only provides a conceptual explanation of the original results, but also closes the gap between the known upper and lower bounds.

The paper is organized as follows. Section~\ref{sec:model} defines the operational and logical models of interest, and provides preliminary notions for epistemic logic. In Section~\ref{sec:task solv}, the introduced machinery is used to state probabilistic task solvability for epistemic models formally. Section~\ref{sec:information levels intro} showcases applications of the developed framework, by proving the aforementioned one-to-one correspondence and tightening the lower bound for randomized coordinated attack. Finally, Section~\ref{sec:conclusion} concludes our paper and discusses possible future work.

\subsection*{Additional Related Work}

Over the last decades, several mathematical frameworks for reasoning about distributed task solvability have been developed. Approaches based on combinatorial topology~\cite{DBLP:conf/stoc/HerlihyS93}, point-set topology~\cite{DBLP:journals/ipl/AlpernS85}, and epistemic logic~\cite{DBLP:journals/jacm/FaginHV92} have all led to important insights into the limits of distributed computation. More recently, the seminal work of Goubault et al.~\cite{GoubaultLR21} bridges topological models and epistemic logic, by relating task solvability to dynamic epistemic logic~\cite{vanDitmarsch2007-VANDEL-6}. In their work, dynamic epistemic logic is used
to model how the knowledge of processes evolves through communication. Inspiration for this approach was earlier work by Casta{\~{n}}eda et al.~\cite{CastanedaDRV24}, which looks at communication patterns epistemically.

These works focus on deterministic algorithms for distributed tasks. In contrast, randomized distributed computation introduces an additional probabilistic dimension, both through randomized process behavior and through uncertainty induced by communication failures. In this work, we focus on randomized process behavior, whereas stochastic communication failures have recently been studied by Fraignaud et al.~\cite{DBLP:conf/wdag/FraigniaudPR25}.

We stress that we consider the case of bounded computation, which is different from the well developed research line on randomized algorithms for consensus with no agreement error, but requiring termination only with high probability  e.g.~\cite{AWbook} and~\cite[Chapter 24]{aspnes2025notestheorydistributedsystems}.

Knowledge based on probabilistic actions has been studied using probabilistic dynamic logic~\cite{ProbabDEL-Kooi-2003}, as well as in the famous \emph{interpreted systems model}   (see~\cite{FHMV1995}), by Halpern~\cite{DBLP:journals/jacm/FaginH94,Halpern2003-HALRAU}.

Building on the approach of Goubault et al.~\cite{GoubaultLR21}, communication patterns~\cite{CastanedaDRV24}, probabilistic epistemic semantics~\cite{ProbabDEL-Kooi-2003,DBLP:journals/jacm/FaginH94,Halpern2003-HALRAU}, and the operational model of Varghese and Lynch~\cite{VargheseL96}, we develop a novel probabilistic epistemic framework for randomized distributed computation. The framework allows us to analyze problems such as randomized coordinated attack through notions of indistinguishability, iterated knowledge, and probabilistic uncertainty.



\section{Operational and logical models}\label{sec:model}
We begin by introducing the model from Lynch~\cite[Chapter 5]{Lynch96}, which we call the \emph{operational model}, generalizing to arbitrary message delivery patterns, following recent research on dynamic networks, but extended with probabilities. Then  we introduce a corresponding \emph{epistemic model}, inspired by the frameworks~\cite{GoubaultLR21,CastanedaDRV24}, extending it with probabilities, following~\cite{ProbabDEL-Kooi-2003}. We also present an epistemic notion of probabilistic task solvability.
A main goal of this section is to show that the operational model and the epistemic model are in a sense equivalent.

\subsection{Operational model}\label{subsec:operational model}
For a positive integer $k\in \mathbb{N}$, we define $[k]:= \{1,\ldots, k\}$, and if $a = (a_i)_{i\in [k]}$ is a sequence then $a_i$ denotes its $i$-th element. Throughout this work, the set $N$ is a set of $n$ processes, denoted
$p_1,\ldots, p_n$, or simply $p,q,$ and $r$ if no ambiguity arises. We use $v_{p_i}$ for the component of a vector $v$ associated with process $p_i$. Before starting a protocol, processes are assigned a private input value from a domain $D$ by an \emph{input configuration} $I: N \rightarrow D$.
In this paper, we consider binary inputs, i.e., $D=\{0,1\}$. 
The set of all input configurations is denoted by $\CI$.\\

\noindent\textbf{\textbf{Communication and protocol.}}
A communication protocol consists of~$R$ synchronous rounds, where at the beginning of each round every process sends a message to all other processes. Messages sent in a given round are simultaneously received at the end of that round, or are lost otherwise.
Throughout this work, we only consider
\emph{full information protocols} with \emph{perfect recall}. That is,  processes send their local state in each round, and remember everything they have received.
We denote $p$'s message in round $r$ by $m_{p,r}$.

A \emph{communication graph} is a directed graph whose vertex set corresponds to the set of processes~$N$. An edge $(p,q)$ in a communication graph indicates that the message sent by process $p$ is delivered to process $q$. The message $m_{p,r}$ is always received by process $p$. Consequently, every vertex has a loop. Given a communication graph $G := (N,E)$, we define $\mathsf{in}_p(G):= \{q\in N \mid (q,p)\in E\}$ as the set of processes from which $p$ receives messages.
A set of communication graphs is called a \emph{communication pattern}.\\

\noindent\textbf{Adversary.}
The adversary $\adv$ chooses inputs and determines which messages are delivered in each round by selecting a communication graph for every round, from a fixed set of graphs. 
This corresponds to an \emph{oblivious} adversary, as the set of possible graphs in each round is the same, and the adversary can select \emph{any} graph from this set, e.g.~\cite{WinklerPGSS24}.
In the usual setting of the coordinated attack problem, the adversary can select any directed graph, including the one where no messages are delivered to other processes.\\

\noindent\textbf{Executions.}
An execution models the evolution of the local states of processes over $R$ rounds,
in a full information protocol. These do not include decision values, which we model later on.
We distinguish between \emph{deterministic} and \emph{probabilistic executions}. 

In a deterministic execution, processes start in an input configuration $I\in \CI$, and the adversary $\adv$ drops messages according to a sequence of communication graphs $G :=  (G_r)_{r\in [R]}$, where $G_r := (N,E_r)$, denoting the set of successful channels at each round. Deterministic executions can thus be parametrized by a tuple $\alpha = (I,G)$.
The set of messages received by process~$p$ at the end of round $r$ is
\[
\rec_p(G,r) := \{m_{q,l} \mid (q,p) \in E_l \text{ for some } l\leq r \}.
\]
Since processes have {perfect recall},
 the local state of a process $p$ in execution $\alpha$ at the end of round $r$ is
\[
s_p(\alpha, r) := (I(p),\rec_p(G,r)).
\]
Moreover, because we only consider
full information protocols, processes send their local state in each round, i.e., $m_{p,r} = s_p(\alpha, r-1)$.

In a probabilistic execution, processes generate randomness by tossing coins. Since we assume full information protocols, the only randomization in their behavior is to include the outcome of random coins in messages, and at the end, the decision of a process may depend on the random coin values that the process has seen in the past
(the number of rounds executed is fixed a priori, and cannot be affected by the randomness). 
We define the set of all probabilistic executions to be $\Omega$. 

We stress that we assume the random coin values are not observable by the adversary.
Thus, there is a probability distribution over executions, once the adversary fixes inputs and communication graphs. 

The local randomness of a process $p$ is a random vector $\bar{X}_p = (X_{p,1},\ldots,X_{p,R})$, where each $X_{p,r}$ is an independent random variable with finite support sampled by process~$p$ in round~$r$. Assuming finite support reflects the standard computational interpretation of randomized distributed algorithms, where processes generate randomness using finite random bit strings (e.g., coin flips).
We define $\bar{X}_{p\mid r} := (X_{p,1}, \ldots, X_{p,r})$, and write $\bar{x}_p$ and $\bar{x}_{p\mid r}$ for realizations of $\bar{X}_p$ and $\bar{X}_{p\mid r}$ respectively. We use $\mathbf{X}$ for the joint variable $(\bar{X}_{p_1}, \ldots , \bar{X}_{p_n})$. Intuitively, $\mathbf{X}$ captures the randomness underlying an execution.

We define a probability space
$\CO = (\mathbf{X},2^\mathbf{X},\mu)$, where $\mu: 2^\mathbf{X} \rightarrow [0,1]$ is a probability measure. 
A probabilistic execution with respect to $\CO$ is a triple $\alpha = (I,G,x)$, where $(I,G)$ is a deterministic execution and $x$ is an element of $\mathbf{X}$. If the adversary $\adv$ chooses $I$ and $G$, the probability that the execution is $(I,G,x)$ equals $\mu(\{x\})$. Moreover, conditioned on $I$ and $G$, probabilities sum up to 1, i.e.,
\[
\sum_{\forall x\in \mathbf{X.}(I,G,x)}\mathsf{Pr}[(I,G,x)] =\sum_{x \in \mathbf{X}} \mu(\{x\}) = 1. 
\]

For a probabilistic execution $\alpha = (I,G,x)$, the local state of process $p$ at the end of round~$r$ additionally contains the sampled randomness over rounds:
\[
s_p(\alpha, r) := (I(p),\bar{x}_{p\mid r}, \rec_p(G,r)).
\]

\noindent\textbf{Decisions.}
At the end of an execution, processes output a decision value. 
Let $\Omega$ be the set of all executions. For each $p\in N$, we define
\[
S_p := \{ s_p(\alpha, R) \mid p\in N \text{ and } \alpha \in \Omega\},
\]
as the set of all possible local states at the end of executions in $\Omega$. We model decisions of processes as a function $\delta_p : S_p \rightarrow \{0,1\}$ and define $\delta$ to represent the decision of all processes.
For fixed $I$ and $G$, this induces a probability distribution on decisions for each process $p$. For example, conditioned on $I$ and $G$, the probability that process $p$ decides the value 1 is
\[
\mathsf{Pr}[\{\alpha \in \Omega \mid \alpha = (I,G, \cdot) \text{ and }\delta_p(s_p(\alpha, R))=1\}].
\]
To simplify notation, we write 
$\mathsf{Pr}_{I,G}[p \text{ decides } 1]$ for the above expression. Lemma~\ref{lem:indist_preserves_prob} is a standard result and states that the probability of a process deciding 1 (or 0) is the same in executions it cannot distinguish (the proof can be found in~\cite{Lynch96}). 

\begin{lemma}\label{lem:indist_preserves_prob}
    Let $I,I' \in \CI$, $G,G'$ be sequences of communication patterns, $p\in N$, and $x,x'\in \mathbf{X}$ with $\bar{x}_p = \bar{x}'_p$ such that $s_p((I,G,x),R) = s_p((I',G',x'),R)$.
    We find that
    \[
    \mathsf{Pr}_{I,G}[p \text{ decides } 1]=\mathsf{Pr}_{I',G'}[p \text{ decides } 1].
    \]
\end{lemma}

Example~\ref{ex:lynch1} shows how our operational model depicts probabilities. The mentioned algorithm (here only outlined) will be explained in Section~\ref{sec:information levels}.

\begin{example}\label{ex:lynch1}
    Varghese and Lynch~\cite{VargheseL96} propose a simple algorithm to solve the randomized coordinated attack problem in $n$ rounds. Before sending its first message, a designated process $p$ rolls a fair $n$-sided die (the so-called \emph{key} $k$) and appends the value to its message. 
    
    Throughout the protocol, processes compute their \emph{information level} (formalized later), which is an integer and depends only on the adversarial choice of $I$ and $G$. It can be shown that processes disagree only if $k$ is equal to a specific level among processes. Consequently, since the adversary fixes $I$ and $G$, it holds that  
    \[
    \mathsf{Pr}_{I,G}[\text{ processes disagree } ] = \frac{1}{n}.
    \]
\end{example}

\subsection{Logical model}\label{subsec:logical model}

This section will introduce logical concepts to later reason about the operational model introduced in Section~\ref{subsec:operational model}. We first introduce the logical language used throughout the paper, and then build toward the introduction of \emph{probabilistic epistemic models}. \\

\noindent\textbf{Language.} The set of propositional variables is $\prop$.
For $p\in \prop$ and $i\in N$, the language~$\CL$  is inductively defined by the following grammar:
\[
\phi ::= p \mid  \lnot  \phi    \mid (\phi \land \phi) \mid K_p \phi.
\]
The formula $K_p\phi$ means that process $p$ \emph{knows} $\phi$. We also define the \emph{everybody knows} modality 
\[
E\phi := \bigwedge_{p\in N}K_p \phi.
\]
We set $E^0\phi := \phi$ and $E^n\phi :=  E(E^{n-1}\phi)$.
In addition to knowledge, we also define a \emph{probabilistic modality} $\SFP_{\leq s}\phi$  that denotes that the probability of $\phi$ being true is a most $s$.
We stress that $\SFP_{\leq s}$ is a global operator, meaning that probabilities are not assigned by processes, but by the adversary (or the environment) instead. This is in line with the operational probability space $\CO$, which induces probabilities over executions from the global perspective of the adversary. Thus, $\SFP_{\leq s}\phi$ means that the adversary believes that $\phi$ is true with at most probability $s$. For convenience, we abbreviate $\SFP_{> s}\phi,\SFP_{< s} \phi, \SFP_{\geq s}$, and $\SFP_{=s}\phi$, which can all be defined in terms of $\SFP_{\leq s} \phi$.
The extension of $\CL$ with the modality $\SFP_{\leq s}$ is denoted by~$\CL_\SFP$. \\

\noindent\textbf{Kripke frames.} To model what processes and the adversary can distinguish during an execution, we use Kripke frames. Formally, let $W$ be a set of worlds (global states) and $A$ be a set of agents (e.g., processes or the adversary). A Kripke frame is a pair $(W,(\sim_a)_{a\in A})$ where each~$\sim_a$ is an equivalence relation on $W$, also called an \emph{indistinguishability} relation. We write $[w]_{a}$ for the equivalence class of $w$ under $\sim_a$. For a set of worlds $W$, we differentiate between two frames:
\begin{enumerate}
    \item the \emph{process-frame} $F_N = (W,(\sim_p)_{p\in N})$; and 
    \item the \emph{adversarial frame} $F_\adv = (W,\sim_\adv)$.
\end{enumerate}
The frames are interpreted as follows; if $w\sim_p v$, then process $p$ cannot distinguish between the worlds $w$ and $v$; if $w \sim_\adv v$, then the adversary $\adv$ cannot tell the two worlds apart. A \emph{Kripke model} is a triple $\CM = (F_N, F_\adv, V)$, where $F_N$ and $F_\adv$ contain the same set of worlds and $V: W \rightarrow 2^\prop$ is a valuation. Adding a valuation $V$ to a frame enables us to reason about what agents actually know.\\

\noindent\textbf{Probabilistic models.}
Consider an arbitrary Kripke model $\CM = (F_N,F_\adv, V)$.
We equip each world of the adversarial Kripke Frame $F_\adv$ with a probability space 
\[
\CO_w := ([w]_\adv, 2^{[w]_\adv}, \mu_w) \text{ such that } \CO_w = \CO_v \text{, if } w \sim_\adv v.
\]
We call the frame $\mathbb{F}_\adv = (W,\sim_\adv, (\CO_w)_{w\in W})$ a \emph{probabilistic adversarial frame}. Furthermore, notice that this choice of modeling implies that
\[
\sum_{v\in [w]_\adv}\mu_w(v)= 1.
\]

An \emph{$\adv$-probabilistic model} is a triple $\CM = (F_N,\mathbb{F}_\adv,V)$, where $V:W\rightarrow 2^\prop$ is a valuation. Definition~\ref{def:relational_truth} illustrates the notion of truth in $\adv$-probabilistic models. Notice that probabilities are always evaluated relative to an adversarial equivalence class. Hence, probabilistic formulas express probabilities relative to the adversary's uncertainty.

\begin{definition}[$\models$]\label{def:relational_truth}
Let $\CM := (F_N, \mathbb{F}_\adv,V)$ be an $\adv$-probabilistic model. For every world $w\in W$
we define the relation $\CM,w \models \phi$ by induction on $\phi$:
\begin{align*}
  & \CM,w \models  p        &\text{if{f}}\qquad & p \in V(w)\\
    &\CM,w \models  \neg \phi                     &\text{if{f}}\qquad  & \CM,w  \not \models \phi \\
  & \CM,w \models \phi \land \psi   &\text{if{f}}\qquad  &\CM,w \models \phi \text{ and } \CM, w\models \psi\\
      &\CM,w\models  K_p \phi \qquad &\text{if{f}} \qquad & \forall v \in [w]_p. \CM,v \models \phi.\\
      &\CM,w \models  \SFP_{\leq s} \phi \qquad &\text{if{f}} \qquad & \mu_w(\{v \in [w]_\adv \mid  \CM, v \models \phi\}) \leq s.
\end{align*}
\end{definition}

Let $\CM = (F_N,\mathbb{F}_\adv,V)$ be an $\adv$-probabilistic model. 
A formula $\phi \in \CL_\SFP$ is \emph{valid in} $\CM$, denoted $\CM \models \phi$, if it is true in all worlds. 

\subsection{Translation}\label{sec:translation}
This section introduces the necessary machinery to picture the operational model logically. Besides standard communication updates we also define updates that model the generation of local randomness. Readers familiar with the literature on dynamic epistemic logic will observe that we do not define the updates based on \emph{action models}. This is purposely done as to not overload the presentation with additional notation. Details on communication updates can be found in~\cite{CastanedaDRV24}. Action models for the generation of randomness are straightforward (as we will demonstrate), because computing probabilities is not part of the update itself.\\

\noindent\textbf{Communication update.} Let $F_N = (W,\sim_p)_{p\in N}$ be a process frame and consider a communication pattern $G = \{G_1, \ldots, G_R \}$. We define the \emph{communication update}:
\[
F_N \otimes G := (W^\otimes, \sim^\otimes_p)_{p\in N},
\]
where $W^\otimes := \{(w,G_i) \mid w\in W \text{ and } G_i \in G\}$,
where $(w,G_i) \sim^\otimes_p (w',G_j)$ if and only if $w\sim_p w'$ and $\textsf{in}_p(G_i) = \textsf{in}_p(G_j)$. In other words, a process $p$ cannot distinguish two worlds $(w,G_i)$ and $(w,G_j)$ if and only if it cannot distinguish $w$ from $w'$ and receives messages from the same set of processes in $G_i$ and $G_j$.

\begin{example}\label{ex:msg_upd}
In the \emph{immediate snapshot model} $\mathcal{IS}$ processes communicate by using snapshots and at least one message gets received. Figure~\ref{fig:ex:msg_upd} depicts the update for two processes $p$ and~$q$, and one global state $w$. In the frame on the right $\top\bot$ means that $p$ received $q$'s message but $p$'s message to $q$ was lost. This scenario corresponds to the communication graph $G_{p} :=\{(q,p), (p,p), (q,q)\}$. Moreover, $p$ cannot distinguish $\top\bot$ from $\top\top$, which corresponds to the communication graph $G_{pq} := \{(p,q),(q,p),(p,p), (q,q)\}$. Formally, $p$ cannot distinguish the worlds induced by $G_p$ and $G_{pq}$ because $\mathsf{in}_p(G_p) =\mathsf{in}_p(G_{pq}) = \{p,q\}$.
\begin{figure}[h]\label{fig:ex:msg}
    \centering
    \begin{tikzpicture}[node distance=7.5mm,  world/.style={
        circle,
        draw,
        minimum size=5mm,
        fill=gray!20,
        inner sep=1pt,
        font=\small
    }]

\node[world] (w) {$w$};
\node[world] (a) [right=3cm of w] {$\top\bot$};
\node[world] (b) [right=of a] {$\top\top$};
\node[world] (d) [right=of b] {$\bot\top$};

\path[-] (a) edge node[above] {$p$} (b);
\path[-] (b) edge node[above] {$q$} (d);

\draw[->, line width=1.5pt, >=Latex, shorten >=6mm, shorten <=6mm]
    (w)
    --
    (a)
    node[midway, above] {update};

\end{tikzpicture}
    \caption{A simple message update for $\mathcal{IS}$.}
    \label{fig:ex:msg_upd}
\end{figure}
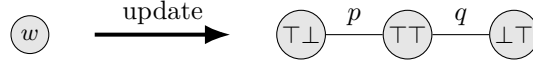

\end{example}

\noindent\textbf{Configuration update.} A process frame $F_N = (W,\sim_p)_{p\in N}$ can be updated with a set $\CV$ of vectors each containing $n$ values. We define the \emph{configuration update} of $F_N$ with $\CV$ as:
\[
F_N \circ \CV := (W^\circ, \sim^\circ_p)_{p\in N},
\]
where $W^\circ := \{(w,v) \mid w\in W \text{ and } v\in \CV\}$ and $(w,v)\sim^\circ_{p_i} (w',v')$ if and only if $w\sim_{p_i} w'$ and $v_{p_i} = v'_{p_i}$. Example~\ref{ex:config_update} shows an update involving two processes, each privately tossing a coin. 

\begin{example}\label{ex:config_update} Consider two processes $p$ and $q$ and an initial global state $w$. Each process privately tosses a coin. Without any communication between the processes, this results in four possible global states. Since each process observes the outcome of its own coin toss, it only considers global states possible in which its local outcome is the same.

For example, in Figure~\ref{fig:ex:cointoss} in the global state $TT$, both processes observe tails. Process $p$ knows that its own outcome is tails, but does not know the outcome observed by $q$. Hence, from the perspective of $p$, both $TT$ and $TH$ are possible. 

\begin{figure}
    \centering
    \begin{tikzpicture}[node distance=7.5mm,  world/.style={
        circle,
        draw,
        minimum size=5mm,
        fill=gray!20,
        inner sep=1pt,
        font=\small
    }]

\node[world] (a) {$TT$};
\node[world] (b) [right=of a] {$HT$};
\node[world] (d) [below=of a] {$TH$};
\node[world] (e) [below=of b] {$HH$};

\node[world] (i1) at ($(a)!0.5!(d)+(-4cm,0)$) {$w$};

\path[-] (a) edge node[above] {$q$} (b);
\path[-] (b) edge node[right] {$p$} (e);
\path[-] (e) edge node[below] {$q$} (d);
\path[-] (a) edge node[left] {$p$} (d);

\draw[->, line width=1.5pt, >=Latex]
    ($(a)!0.5!(d)+(-3cm,0)$)
    --
    ($(a)!0.5!(d)+(-1cm,0)$)
    node[midway, above] {update};

\end{tikzpicture}
    \caption{A simple configuration update for two private local coin tosses.}
    \label{fig:ex:cointoss}
\end{figure}
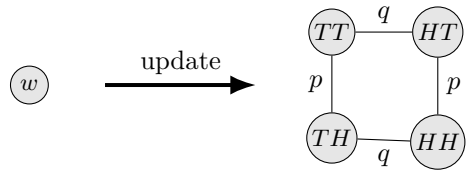

\end{example}

\noindent\textbf{$N$-execution Frame.} We now build a frame that corresponds to the execution from the perspective of a process. The constructions starts from an \emph{input frame} for private inputs $I = (\CI, \sim_p)_{p\in N}$, where $I_1 \sim_p I_2$ if and only if $I_1(p) = I_2(p)$, and iteratively applies configuration and message updates. 

Let $(G_r)_{r\in [R]}$ be a sequence with the communication patterns for each round. Further for $r\in [R]$, let $\mathbf{X}_r$ be the set containing all possible realizations of the joint variable $(X_{p_1, r}, \ldots X_{p_n, r})$.
The \emph{$N$-execution frame} $F^\CA_N$ is iteratively constructed as follows:
\[
    F_1 := (I \circ \mathbf{X_1})\otimes G_1,\;
    F_2  := (F_1 \circ \mathbf{X_2})\otimes G_2 ,\;
    \ldots,\;
    F^\CA_N := (F_{R-1} \circ \mathbf{X_R})\otimes G_R.
\]

We use the superscript $\CA$ to be consistent with literature on task solvability and dynamic epistemic logic. Historically, $\CA$ is used to emphasize the action model, which we do not model explicitly for compactness.

We write $(I,G,x)$ for the worlds of $F^\CA_N$ where $I$ is the initial input configuration, $G$ is the sequence of communication graphs applied to $I$, and $x$ is the sampled randomness. Lemma~\ref{lem:correspnence} states the relation between $F^\CA_N$ and the set of all possible probabilistic executions $\Omega$.

\begin{lemma}\label{lem:correspnence}
    There is a one-to-one correspondence between the set of all probabilistic executions $\Omega$ and the $N$-execution frame $F^\CA_N$. Furthermore, a process $p$ cannot distinguish two worlds $w = (I,G,x)$ and $w' = (I',G',x')$, denoted $w\sim^\CA_p w'$ if and only if it has the same local state in both executions, i.e., $s_p(w,R) = s_p(w', R)$.
\end{lemma}
\begin{proof}
    The proof is standard (see for example~\cite{FHMV1995,CastanedaDRV24}).
\end{proof}

\noindent\textbf{$\adv$-execution frame.}  
We can build a probabilistic adversarial execution frame~$F^\CA_\adv$ from the worlds of $F^\CA_N$. The adversarial indistinguishability relation is given by:
\[
(I,G,x)\sim^\CA_\adv (I',G',x') \quad \text{iff.}\quad I=I' \text{ and } G = G'.
\]
The probability space $\CO = (\mathbf{X},2^\mathbf{X},\mu)$ of the operational model, naturally induces a probability space $\CO_\alpha$ for every probabilistic execution $\alpha = (I,G,x) \in \Omega$,
which defines a probabilistic adversarial frame. Here, the probability measure must be conditioned on being in worlds of~$[\alpha]_\adv$. The measure normalizes correctly, if the operational measure normalizes. Formally, we set $\CO_\alpha := ([\alpha]_\adv, 2^{[\alpha]_\adv},\mu|[\alpha]_\adv)$, where $\mu|[\alpha]_\adv$ is the operational measure $\mu$ conditioned on the adversarial choice of $I$ and $G$.

Constructing $F^\CA_\adv$ after the configuration update is possible because the operational model already encodes randomness via realizations of the joint random variable $\mathbf{X}$. 
In more general dynamic epistemic settings, probabilistic updates modify probability distributions during the update process itself. This requires genuinely dynamic probabilistic update mechanisms as developed in~\cite{ProbabDEL-Kooi-2003,vanBenthemGerbrandyKooi2009}.\\

\noindent\textbf{Execution model.} 
The set $\prop$ contains all propositional variables. The variables describe either inputs or realizations of a random variable $X_{p,r}$. The propositional variable $i_p$ represents process $p$'s binary input, and $x^s_{p,r}$ encodes the randomness of process $p$ in round~$r$. Notice, that~$s$ depends on the support of the random variable $X_{p,r}$, denoted by $\textsf{supp}(X_{p,r})$. 
Formally, the set propositional variables is:
\[
\prop := \{i_p \mid  p\in N\} \cup \{x^s_{p,r} \mid r\in[R],p\in N \text{ and } s \in \textsf{supp}(X_{p,r})\}.
\]
The \emph{execution model} is a triple $\CI[\CA] = (F^\CA_N, F^\CA_\adv, V^\CA)$, where $V^\CA$ is a valuation such that:
\begin{enumerate}
    \item $i_p \in V^\CA((I,G,x)))$ if and only if $I(p) = 1$;
    \item $x^s_{p,r} \in V^\CA((I,G,x))$ if and only if $x_{p,r} = s$.
\end{enumerate}
Again, we use $\CI[\CA]$ to adhere to previous literature. The notation emphasizes that the execution is built from the set of possible inputs by applying updates.

\section{Task solvability}\label{sec:task solv}
We now provide a novel notion of probabilistic task solvability for epistemic models. This definition does not depend on the operational model for randomized coordinated attack, and characterizes the solvability of various randomized tasks.
We follow the standard approach of relating task solvability to the existence of a map from inputs to admissible outputs. \\

\noindent\textbf{Task model.} Let $\CT$ be a task. The \emph{task model} $\CI[\CT] = (F^\CT_N, F^\CT_\adv, V^\CT)$ is built over a set of output worlds~$W^\CT$, where each world is of the form $(I, G, o)$. Here, $o$ is an output vector compatible with the input configuration $I\in \CI$. Two worlds are indistinguishable for a process if and only if they have the same outputs and inputs, i.e., $o_p = o'_p$ and $I(p) = I'(p)$. The valuation $V^\CT$ assigns propositional variables denoting the inputs and outputs as in Section~\ref{sec:translation}. Encoding the delivered messages is necessary for modeling the adversary. The adversarial indistinguishability relation is as before. We stress that the model $\CI[\CT]$ is deterministic. Similar to the operational model, probabilities are induced by a decision function.\\

\noindent \textbf{Decisions.} Given an execution model $\CI[\CA]$, we model the decision of processes as a function $\delta : W^\CA \rightarrow \{0,1\}^n$ from the worlds of the execution model to a binary output vector. The only condition is that processes decide the same values in indistinguishable global states, i.e.,
for all processes $p\in N$, we have $w\sim^\CA_p w'$ implies that $\delta(w)_p = \delta(w')_p$. 

Given a decision function $\delta$, we transform the adversarial task model $F^\CT_\adv$ into an $\mathbb{A}$-probabilistic output frame $\mathbb{F}^O_\adv$. The set of worlds for $\mathbb{F}^O_\adv$ is
\[
W^O := \{(I,G,o) \in W^\CT \mid \exists (I,G,x)\in W^\CA.\delta((I,G,x)) = o\}.
\]
In words, the frame $\mathbb{F}^O_\adv$ is formed over the worlds of $F^\CT_\mathbb{A} $that are actually produced by $\delta$. This is needed in order to properly define a probability measure for each world. Notice, that this deviates from the deterministic notion of task solvability in~\cite{GoubaultLR21}.

We define a probability measure~$\mu_{w^*}^O$  for each world 
 $w^* := (I,G,\cdot) \in W^O$ and event $E\subseteq [w^*]_\adv$, in terms of the operational measure conditioned on the choice of $I$ and $G$:
 \[
 \mu_{w^*}^O(E) = \mu(\{ x\in \mathbf{X} \mid  (I,G,\delta(I,G, x))\in E\}).
 \]

The probability space for each world $w^*\in W^O$ is then defined as before.
The measure~$\mu^O$ normalizes correctly because the preimages induced by $\delta$ partition the worlds of $\CM^\CA$, and the operational measure $\mu$ sums to $1$ on every adversarial equivalence class. The output model $\CM^O := (F^O_N,\mathbb{F}^O_\adv,V^O)$ is obtained by restricting $F^\CT_N$ to $W^O$ as well, and defining an \emph{output valuation} $V^O$ such that
\[
\CM^O, (I,G,\bar{o}) \models o_p \quad \text{iff.}\quad  \bar{o}_p = 1,
\]
where the propositional variable $o_p$ represents $p$'s binary decision.
Intuitively, the output model abstracts only the final decisions relevant for task specifications.\\

\begin{figure}
    \centering
    \begin{tikzpicture}[node distance=7.5mm]

\node at (0,0) (a) {$\CI$};
\node at (7.5,3) (b) {$\CM^O := (F^O_N, \mathbb{F}^O_\adv, V^O)$};
\node at (0,3) (c) {$\CI[\mathcal{A}]:= (F^\CA_N, F^\CA_\adv, V^\CA)$};
\node at (7.5,0) (d) {$\CI[\CT]:= (F^\CT_N, F^\CT_\adv, V^\CT)$};

\draw[->] (d) -- node[midway, right] {implies $V^O$} (b);
\draw[->] (a) -- node[midway, below] {task specification yields $\CI[\CT]$} (d);
\draw[->] (a) -- 
  node[midway, left, align=center]
  {communication and\\configuration updates}
(c);
\draw[->, dashed] (c) -- 
  node[midway, above, align=center]
  {$\delta$ induces $(F^O_N, \mathbb{F}^O_\adv)$}
(b);

\end{tikzpicture}
    \caption{Illustration of probabilistic epistemic task solvability. The execution model~$\CI[\CA]$ and the task model $\CI[\CT]$ can be constructed based on the input frame $\CI$. The task model together with a decision function $\delta$ induce the output model $\CM^O$. A task is solvable if and only if there exists a decision function $\delta$ such that $\CM^O$ satisfies the logical representation of the task specification.}
    \label{fig:task_solvability}
\end{figure}
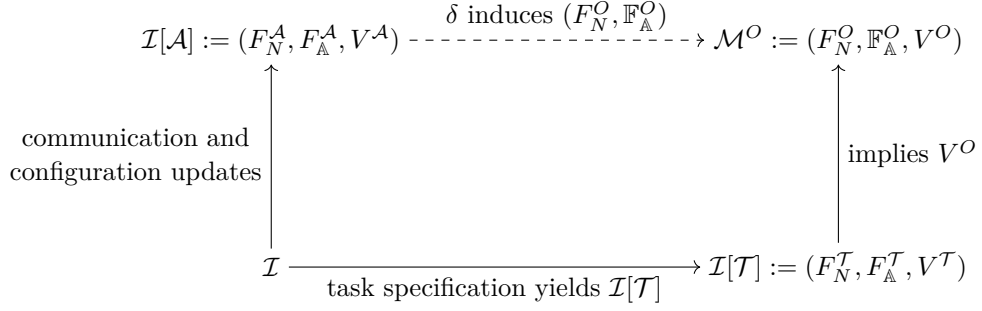

\noindent\textbf{Characteristic formulas.} Given a world $w:=(I,G,x)$ of the execution model $\CI[\CA]$, it is a standard result that there exists a non probabilistic formula $\psi_w \in \CL$ that is exactly true in $w$ and nowhere else~\cite{FHMV1995,Hintikka1962-HINKAB-4}. This formula is called the \emph{characteristic formula of $w$}. Such formulas exist because the sets $W^\CA$ and $\prop$ are finite\footnote{Recall that inputs are binary, random variables have finite support, and the number of rounds is fixed.}, and every two worlds can be distinguished by at least one process\footnote{This is because any two communication graphs can be distinguished by at least one process. The literature refers to this property also as \emph{proper}~\cite{DBLP:conf/stacs/GoubaultLR22,DBLP:journals/corr/abs-2002-08863}}.  

Given a decision function $\delta$, we can use characteristic formulas to describe the sets of worlds, where a process decides 1. The formula $\psi^p_{I,G}$ describes this set conditioned on the adversarial choice of $I$ and $G$. If $\neg \psi^p_{I,G}$ is true in a world $(I,G,x)$, then process $p$ decides 0. We use the characteristic formula $\psi^\mathsf{Dis}_{I,G}$ to encode the worlds where the adversary's choice of  $I$ and $G$ was successful, i.e., there exist processes $p$ and $q$ such that $\psi^p_{I,G} \land \neg \psi^q_{I,G}$ is true.

Since all our logical models are i) based on a measurable operational model, and ii) only condition the operational measure based on equivalence classes,  we can relate the probability of a process deciding 1 in the output model $\CM^O$ to the probability of deciding 1 in the operational model:
\[
\CM^O,(I,G,o) \models \SFP_{\leq s}o_p \quad \text{iff.}\quad \CI[\CA] \models \SFP_{\leq s}\psi^p_{I,G} \quad \text{iff.}\quad \prob_{I,G}[p \text{ decides }1]\leq s
\]
The middle statement is true in all worlds because if a world $w'$ belongs to a different adversarial equivalence class, e.g. $w'\notin[(I,G,\cdot)]_\adv$, then $\psi^p_{I,G}$ is false in every world of $[w']_\adv$, and therefore has probability 0.\\

\noindent\textbf{Task solvability.} For deterministic tasks, there is one input model and one output model. However, for randomized tasks, there exists one set of feasible output worlds, but infinitely many output models that satisfy probabilistic properties. Since characteristic formulas exist, the properties of a task $\CT$ can be expressed as logical formulas, called $\CT$-formulas. Given an execution model $\CI[\CA]$, we say that the decision function $\delta$ solves a task~$\CT$ if and only if every world of the model $\CM^O$ induced by $\delta$  satisfies the $\CT$-formulas. More formally, for any $\CT$-formula $\tau$, it holds that $\CM^O \models \tau$. Figure~\ref{fig:task_solvability} illustrates task solvability. 

\begin{remark}
For deterministic tasks, the presented notion of task solvability is closely related to the one introduced in~\cite{GoubaultLR21}. While the underlying logical frameworks differ, we conjecture the notions are equivalent, because $I[\CT]$ fully characterizes $\CM^O$ in the deterministic setting. 
\end{remark}
\section{Applications}\label{sec:information levels intro}

Section~\ref{sec:model} introduced a probabilistic epistemic framework for randomized distributed computation by relating executions, knowledge, and probabilistic outputs. We now apply this framework to randomized coordinated attack. First, we show that the information levels introduced by Varghese and Lynch correspond exactly to finite levels of iterated knowledge. We then use the same epistemic structure to explain the lower bound: decision probabilities can only change gradually between indistinguishable executions, and executions with no communication and full communication are connected by such indistinguishability chains.

\subsection{Information levels}\label{sec:information levels}

Given a probabilistic execution $\alpha := (I,G,x)$, we informally define the \emph{information level} of a process $p$ inductively:
\begin{description}
    \item[Base case:] At the beginning of round 1 all processes have level 0;
    \item[Induction:] If $p$ receives messages indicating that the minimum level among all processes is~$k$, then $p$ sets its level to $k+1$.
\end{description}

The formal definition of information levels, given in ~\cite{Lynch96}, can be found in the Appendix.
We denote the level of process $p$ in execution $\alpha$ in round $r$ with $\ell(p,\alpha, r)$. Theorem~\ref{thm:epistemic levels} relates the information level of a process at the end of the last round of a probabilistic execution $\alpha := (I,G,x)$ with its knowledge in the world $(I,G,x)$ of the execution model. The expression~$\neg K_pE^n\phi_\alpha$ ensures that processes have \emph{exactly} level $n$.

\begin{theorem}\label{thm:epistemic levels}
    Let $\CI[\CA]$ be the execution model. For every $\alpha := (I,G,x) \in \Omega$, we find that 
    \[
    \ell(p, \alpha,R)=n\quad \text{if{f}.} \quad \CI[\CA],\alpha \models K_pE^{n-1}\phi_\alpha \land\lnot K_pE^n\phi_\alpha,
    \]
    where $\phi_\alpha$ encodes the inputs in $\alpha$.
\end{theorem}
\begin{proof}
    See Appendix.
\end{proof}

We can relate levels to the validity condition of randomized coordinated attack. Recall that validity ensures that if all messages are delivered and all processes have input 1, then they output 1.
Moreover, let $\psi_{\mathsf{Dis}}$ be the formula
 encoding worlds where processes decide differently. For any world $w = (I,G,o)$, the $\CT$-formulas of randomized coordinated attack are
\[
    \psi_{\mathsf{val}} := \left( \bigwedge_{p\in N}i_p \land E^{R}\phi_w\right)\rightarrow \bigwedge_{p\in N} o_p \quad \text{and}\quad 
    \psi_{\mathsf{randA}} := \SFP_{\leq \epsilon}\psi_{\mathsf{Dis}}.
\]
By unfolding the logical operators, we get
\[
E^{R}\phi_w \equiv \bigwedge_{p\in N}K_pE^{R-1}\phi_w,
\]
and by Theorem~\ref{thm:epistemic levels}, this means that all processes have at least level $R-1$ at the end of the execution. Notice that this can only be true if all messages were delivered. We do not need to formulate \textbf{Termination} as it is guaranteed for finite rounds. Example~\ref{ex:semantics} shows how we can use the previously developed machinery to fully characterize randomized coordinated attack in the communication model of Section~\ref{sec:model}.

\begin{example}\label{ex:semantics} This example shows why the $\CT$-formulas $\psi_{\mathsf{val}}$ and $\psi_{\textsf{randA}}$ characterize randomized coordinated attack.
Let $\delta$ be a decision function solving the probabilistic coordinated attack problem. By Lemma~\ref{lem:correspnence}, the execution model $\CI[\CA]$ captures every execution in $\Omega$.
As explained in Section~\ref{sec:task solv}, the decision function $\delta$ together with the task specification produce an output model $\CM^O$.
Since $\delta$ solves the problem, it holds that for any choice of $I$ and $G$, the probability of the adversary's success can be at most some~$\epsilon$, which is logically captured by:
\begin{align*}
\CM^O,(I,G,o) \models \SFP_{\leq \epsilon}\psi_{\mathsf{Dis}} \quad & \text{iff.}\quad \CI[\CA] \models \SFP_{\leq \epsilon}\psi^\mathsf{Dis}_{I,G} \quad \\&\text{iff.}\quad \mu( \{x \in \mathbf{X}  \mid (I,G,x) \text{ violates \textbf{Agreement}}\} )\leq \epsilon.
\end{align*}
Moreover, if $I$ is the input configuration where all processes have input 1 and $G$ is the sequence of communication graphs where all messages are delivered, then  any correct decision function must output only 1's, because the premise of $\psi_{\mathsf{val}}$ is met. Therefore, for any random value $x$, we find for the output model $\CM^O, (I,G,\delta((I,G,x)) \models \SFP_{=1}\bigwedge_{p\in N}o_p$.

\end{example}

We now explain the intuition behind Theorem~\ref{thm:epistemic levels}.
Let $\phi$ be the formula representing the actual input configuration. We now analyze process $p$'s knowledge at levels 1 to 3. \\

\noindent\textbf{Level 1.}
Process $p$ has level 1 if and only if it has heard of every other process. Process $p$ has heard of process $q$ (or $r$) if and only if:
\begin{enumerate}
    \item $p$ receives a message from the corresponding process; or
    \item $p$ receives a message from  process $q$ (or $r$) which has heard of $r$ (or $q$) in an earlier round.
\end{enumerate}
Since a message contains all the inputs the sender knows about, having level 1 is equivalent to knowing all the inputs, i.e., $K_p \phi$. \\

\noindent\textbf{Level 2.}
Process $p$ has level 2 if and only if it knows that every other process has level 1. Epistemically, process $p$ has level 2 if and only if it knows that $q$ and $r$ know all the inputs: 
\[
K_p(K_q\phi \land K_r \phi) \equiv K_pE_{-p}\phi,
\]
where $E_{-p}$ means everybody except $p$ knows.
By construction of our logical model from the operational one, we have that inputs are unique and process $p$ knows its own input. Together with $K_pE_{-p}\phi$, this implies that process $p$ knows that all processes know $\phi$, i.e., $K_pE\phi$.\\

\noindent\textbf{Level 3.}
Process $p$ has level 3 if and only if it knows that every other process has level 2. This is expressed as an epistemic formula below, which reduces to $K_pE^2\phi$.
\[
    K_p(K_qE\phi \land K_rE\phi) \equiv K_p(K_q(K_p\phi \land K_q\phi\land K_r \phi) \land K_r(K_q\phi \land K_p \phi \land K_r\phi)).
\]
In the reduction, we use logical reasoning from $\mathsf{S5}$ and use results from~\cite{Lynch96}. The details of the logical derivation can be found in the Appendix where we prove Theorem~\ref{thm:epistemic levels}.



Theorem~\ref{thm:epistemic levels} allows us to express the algorithm of~\cite{VargheseL96} succinctly and elegantly as a knowledge-based program~(see~\cite{FHMV1995}), shown in Algorithm~\ref{alg:rand_coord_attack} . Formulated as a knowledge-based program, the algorithm becomes surprisingly simple. At the end of round $R$ a process $p$ decides 1 only if it knows that all processes have input 1, e.g. $K_p\phi_{\mathbf{1}}$ , and its iterated knowledge of this fact is at least $k-1$ levels deep. It can be shown that the algorithm's correctness is restricted to the probability~$R^{-1}$. Indeed, from Lynch and Varghese~\cite{VargheseL96}, processes decide differently if and only if one process, say $p$, has level exactly $k$, and another one, say $q$, has level exactly $k-1$. This happens with probability $\frac{1}{R}$. Translated to our model, the following is true:
\[
\CI[\CA] \models\SFP_{=\frac{1}{R}} \left( x^k_{1,1} \land \bigvee_{p,q\in N}(K_pE^{k-1}\phi_{\mathbf{1}} \land \neg K_qE^{k-1}\phi_{\mathbf{1}})\right)
\]

\begin{algorithm}[t]
\caption{A knowledge based program for the algorithm in~\cite{VargheseL96} for process $p_i$. The formula $\phi_{\mathbf{1}}$ denotes that the inputs of all processes are 1.}
\label{alg:rand_coord_attack}
\begin{algorithmic}[1]
\If{$p_i = p_1 \land \mathsf{round}=1$}
    \State $k \gets \mathsf{coinflip}(k)$ \Comment{Generation of the \emph{key} $k$.}
    \State $x^k_{1,1} \gets \mathsf{true}$
\EndIf

\If{$K_{p_i} x^k_{1,1} \land K_{p_i} E^{l}\phi_{\mathbf{1}} \land \neg K_{p_i} E^{l+1} \phi_{\mathbf{1}}$ for some $l\geq k-1$}
    \State $\textbf{decide } 1$
\Else
    \State $\textbf{decide } 0$
\EndIf

\end{algorithmic}
\end{algorithm}
\subsection{Lower bound}\label{sec:lb}
We now show that for any randomized coordinated attack protocol, the probability of disagreeing is at least $R^{-1}$. This bound improves the prior one of $(R+1)^{-1}$ and is tight. The key idea is that indistinguishability in the execution model $\CI[\CA]$ limits how quickly decision probabilities can change in the output model~$\CM^O$. Before stating the result, we briefly recall task solvability from Section~\ref{sec:task solv}.

Consider Figure~\ref{fig:task_solvability} again, and assume a protocol that solves randomized coordinated attack in the operational model with decision function $\delta$. As demonstrated in Section~\ref{sec:translation}, the operational model induces an equivalent execution model $\CI[\CA]$, while the task specification induces the task model $\CI[\CT]$, which has no probabilities. The decision function $\delta$ connects these two structures by assigning outputs to executions, thereby inducing the $\adv$-probabilistic output model $\CM^O$. Intuitively, the output probabilities are inherited from executions. Notice, that if the protocol in the operational model did not solve the task, the induced model $\CM^O$ would not satisfy the task formulas.

We can now state our result and sketch its proof. The main argument is based on a \emph{$pq$-chain}, which is a sequence of executions $\alpha_1, \ldots \alpha_n \in W^\CA$, such that
\[
\alpha_1 \sim^\CA_p\alpha_2 \sim^\CA_q \alpha_3 \ldots \sim^\CA_p \alpha_n.
\]
Intuitively, output probabilities cannot change abruptly between executions along $pq$-chains, which allows us to restrict them from above. As standard in indistinguishability arguments, we start with an execution $\alpha_1$ in which no messages are delivered, and end in an execution $\alpha_n$, where all messages are received. By the properties of randomized coordinated agreement, process $p$ must decide $0$ in $\alpha_1$ and $1$ in $\alpha_n$, if all inputs are 1.

Theorem~\ref{thm:lb} states the lower bound of $R^{-1}$ for the probability of disagreement for randomized coordinated attack. It was proven in~\cite{VargheseL96} that Algorithm~\ref{alg:rand_coord_attack} has a probability of disagreeing of at least $R^{-1}$, which makes our lower bound tight. The proof can be found in the Appendix and we will outline it subsequently. 

The first step towards proving Theorem~\ref{thm:lb} is to establish the standard result stated Lemma~\ref{lem:prob_difference} for the operational model. Informally, it states that the probability of two processes deciding the value 1 can differ by at most $\epsilon$. Intuitively, if the difference exceeds $\epsilon$, then the processes become more likely to disagree, which violates \textbf{Randomized agreement}. This result is a crucial tool when reasoning about output probabilities along a $pq$-chain. As illustrated in Section~\ref{sec:translation}, this result from the operational model directly translates to the epistemic output model.
Corollary~\ref{cor:difference of probabilities} states the result formally in the epistemic setting.

\begin{lemma}\label{lem:prob_difference}
In the operational model, for any protocol that solves randomized coordinated attack, it holds that 
    $|\prob_{I,G}[i \;\operatorname{ decides }]-\prob_{I,G}[j \;\operatorname{ decides }\; 1 ]| \leq \epsilon$.
\end{lemma}
\begin{proof}
    See Appendix.
\end{proof}

\begin{corollary}\label{cor:difference of probabilities}
Let $\CM^O$ be an output model based on a decision function $\delta$ that solves randomized coordinated attack. For every process $p$, if 
$\CM^O, (I,G,o) \models \mathsf{P}_{= s}o_p$, then for all $q\neq p$, it holds that $\CM^O, (I,G,o) \models \mathsf{P}_{\leq s+\epsilon}o_q \land \mathsf{P}_{\geq s-\epsilon}o_q$.
\end{corollary}

We now motivate how probabilities can be bound from above along a $pq$-chain by applying Corollary~\ref{cor:difference of probabilities} to bind probabilities of deciding the value 1.
Let $\alpha_1 := (I_1,G_1,x_1) \in W^\CA$ and $\alpha_2:= (I_2,G_2,x_2) \in W^\CA$ with $\alpha_1 \sim^\CA_p \alpha_2$ be two probabilistic executions. By Lemma~\ref{lem:indist_preserves_prob}, equal local states preserve output probabilities, and thus it holds that:
\[
\prob_{I_1,G_1}[p \text{ decides }1] = \prob_{I_2,G_2}[p \text{ decides }1].
\]
This directly translates to the output model $\CM^O$. Indeed, by construction, we can relate the operational model to the output model $\CM^O$, via the execution model $\CI[\CA]$, because:
\[
\CM^O,w \models \SFP_{=s}o_p \quad \text{iff.} \quad \CI[\CA] \models \SFP_{= s}\psi^1_p \quad\text{iff.}\quad \prob_{I,G}[p \text{ decides }1]= s.
\]
Therefore, it holds that:
\[
\CM^O,(I,G, \delta(\alpha_1)) \models \SFP_{=s}o_p \quad \text{and}\quad \CM^O,(I,G,\delta(\alpha_2))\models \SFP_{=s}o_p.
\]
Consider now the four worlds below, and assume that the probability of $p$ and $q$ deciding the value 1 in $w_0$ is 0:

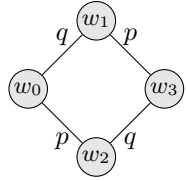
\begin{center}

    \begin{tikzpicture}[node distance=7.5mm,  world/.style={
        circle,
        draw,
        minimum size=4mm,
        fill=gray!20,
        inner sep=1pt,
        font=\small
    }]

\node[world] (a) {$w_0$};
\node[world] (b) [above right=of a] {$w_1$};
\node[world] (d) [below right=of a] {$w_2$};
\node[world] (e) [below right=of b] {$w_3$};

\path[-] (a) edge node[midway,above] {$q$} (b);
\path[-] (b) edge node[midway,above] {$p$} (e);
\path[-] (e) edge node[midway,below] {$q$} (d);
\path[-] (a) edge node[midway,below] {$p$} (d);

\end{tikzpicture}
\end{center}

By applying Corollary~\ref{cor:difference of probabilities} we immediately obtain that process $p$ decides the value 1 with at most $\epsilon$ in $w_1$. By the same reasoning, this is also true for process $q$ in $w_2$. By Lemma~\ref{lem:indist_preserves_prob}, indistinguishability preserves probabilities and we get that $p$ and $q$ need to decide the value~1 with the same probability as in $w_1$ and $w_2$ respectively. Therefore, the individual probabilities of deciding the value 1 are both bound by $\epsilon$. This is extremely helpful, because we found a $pq$-chain ($w_0,w_1,w_3$) along which $p$'s probability of deciding 1 increases by at most $\epsilon$.

Using epistemic reasoning, we can construct a $pq$-chain with $R$-many such segments. In particular, the chain will end in an execution in which $p$ decides 1 with probability 1. In the proof, we use the execution where every process has input 1 and $p$ receives all messages. Thus, $p$'s probability to decide 1 needs to go from 0 to 1 in $R$ many steps of size at most $\epsilon$, yielding $\epsilon \geq R^{-1}$.  Interestingly, this proof only considers two out of $N$ processes, and mostly relies on epistemic indistinguishability arguments. Details are provided in the Appendix.

\begin{theorem}\label{thm:lb}
Any decision function $\delta$ solving randomized coordinated attack induces an output model $\CM^O$ such that for all worlds $w$:
\[
\CM^O, w \models \SFP_{\geq\frac{1}{R}}\psi_{\mathsf{Dis}}.
\]
\end{theorem}
\begin{proof}
    See Appendix.
\end{proof}

Our model also allows us to study randomized coordinated attack for more structured communication patterns.
The \emph{iterated immediate snapshot model} ($\mathcal{IIS}$) is the iterated version of the communication pattern presented in Example~\ref{ex:msg_upd}.
As it turns out, the proof techniques employed for Theorem~\ref{thm:lb} can also be applied here, and we obtain Theorem~\ref{thm:lb_iis}, which says that in the more restricted $\mathcal{IIS}$, a lower bound on $\epsilon$ is $c3^{-R}$ for some constant~$c$. The intuition is that the number of worlds in the updated model grows exponentially with a factor of $3$, and the updated model for two processes forms a $pq$-chain (see~Example~\ref{ex:msg_upd}).

Algorithm~\ref{alg:rand_coord_attack} is still applicable in $\mathcal{IIS}$, although it may no longer be optimal because the processes have access to additional information. In particular, they know that in each round at least one message is received by some process, which may reduce the number of rounds required for a decision or, by Theorem~\ref{thm:epistemic levels}, relax the required knowledge.

\begin{theorem}\label{thm:lb_iis}
Any decision function $\delta$ solving randomized coordinated attack in $\mathcal{IIS}$ induces an output model $\CM^O$ such that for all worlds $w$:
\[
\CM^O, w \models \SFP_{\geq c3^{-R}}\psi_{\mathsf{Dis}},
\]
where $c\in \mathbb{R}$ is a constant.
\end{theorem}
\begin{proof}
    See Appendix.
\end{proof}

\section{Conclusion}\label{sec:conclusion}
In this work we developed a probabilistic epistemic framework for bounded time task solvability. Processes may flip coins to solve a task with probability of error as small as possible in a given number of rounds. We illustrated the use of the framework  focusing on the randomized coordinated attack problem.
Some of the benefits obtained are closing previous upper and lower bounds, as well as  formalizing nested knowledge, a notion that plays an important role in agreement tasks, such as consensus and approximate agreement. 

We emphasize that the presented framework is general and can be used to study randomized algorithms under arbitrary communication patterns, as well as problems other than coordinated attack.
We illustrated the generality 
with a lower bound  in an immediate snapshots model of main interest in shared memory, showing that the same type of argument can be used the coordinated attack setting, as well as in other communication models.

Two main research directions seem particularly interesting. First, to study (bounded) randomized solvability of other tasks. Our techniques apply directly to  
approximate agreement e.g.~\cite{FuggerNS21}, which we will include in the full version. It would be interesting to study other tasks, such as set agreement~\cite{Chaudhuri90}. It would be interesting to study also equality negation, which has been analyzed both topologically~\cite{GoubaultLLR19} and epistemically~\cite{DitmarschGLLR21}.

The second research direction is on the epistemic side, where many interesting problems remain open. In particular, about a sound and complete axiomatization, and the expressivity of the language.
This paper is primarily focused on the formal study of randomized algorithms, however, it would be interesting to consider  more general updates beyond messages and coin flips. In particular, extending the framework to probabilistic actions appears to be an important direction for future work.

\bibliography{ref-probabilistic.bib}

@inproceedings{GoubaultLLR19,
  author       = {{\'{E}}ric Goubault and
                  Marijana Lazic and
                  J{\'{e}}r{\'{e}}my Ledent and
                  Sergio Rajsbaum},
  editor       = {Jukka Suomela},
  title        = {Wait-Free Solvability of Equality Negation Tasks},
  booktitle    = {33rd International Symposium on Distributed Computing, {DISC} 2019,
                  Budapest, Hungary, October 14-18, 2019},
  series       = {LIPIcs},
  pages        = {21:1--21:16},
  publisher    = {Schloss Dagstuhl - Leibniz-Zentrum f{\"{u}}r Informatik},
  year         = {2019},
  url          = {https://doi.org/10.4230/LIPIcs.DISC.2019.21},
  doi          = {10.4230/LIPICS.DISC.2019.21},
  bibsource    = {dblp computer science bibliography, https://dblp.org}
}

@article{DitmarschGLLR21,
  author       = {Hans van Ditmarsch and
                  {\'{E}}ric Goubault and
                  Marijana Lazic and
                  J{\'{e}}r{\'{e}}my Ledent and
                  Sergio Rajsbaum},
  title        = {A dynamic epistemic logic analysis of equality negation and other
                  epistemic covering tasks},
  journal      = {J. Log. Algebraic Methods Program.},
  volume       = {121},
  pages        = {100662},
  year         = {2021},
  url          = {https://doi.org/10.1016/j.jlamp.2021.100662},
  doi          = {10.1016/J.JLAMP.2021.100662},
  bibsource    = {dblp computer science bibliography, https://dblp.org}
}

@inproceedings{Chaudhuri90,
  author       = {Soma Chaudhuri},
  editor       = {Cynthia Dwork},
  title        = {Agreement is Harder than Consensus: Set Consensus Problems in Totally
                  Asynchronous Systems},
  booktitle    = {Proceedings of the Ninth Annual {ACM} Symposium on Principles of Distributed
                  Computing, Quebec City, Quebec, Canada, August 22-24, 1990},
  pages        = {311--324},
  publisher    = {{ACM}},
  year         = {1990},
  url          = {https://doi.org/10.1145/93385.93431},
  doi          = {10.1145/93385.93431},
  bibsource    = {dblp computer science bibliography, https://dblp.org}
}

@book{AWbook,
	address = {Hoboken, NJ, USA},
	author = {Attiya, Hagit and Welch, Jennifer},
	isbn = {0471453242},
	publisher = {John Wiley \& Sons},
	title = {Distributed Computing: Fundamentals, Simulations and Advanced Topics},
	year = {2004}}

@article{WinklerPGSS24,
  author       = {Kyrill Winkler and
                  Ami Paz and
                  Hugo Rincon Galeana and
                  Stefan Schmid and
                  Ulrich Schmid},
  title        = {The Time Complexity of Consensus Under Oblivious Message Adversaries},
  journal      = {Algorithmica},
  volume       = {86},
  number       = {6},
  pages        = {1830--1861},
  year         = {2024},
  url          = {https://doi.org/10.1007/s00453-024-01209-4},
  doi          = {10.1007/S00453-024-01209-4},
  bibsource    = {dblp computer science bibliography, https://dblp.org}
}

@article{GoubaultLR21,
  author       = {{\'{E}}ric Goubault and
                  J{\'{e}}r{\'{e}}my Ledent and
                  Sergio Rajsbaum},
  title        = {A simplicial complex model for dynamic epistemic logic to study distributed
                  task computability},
  journal      = {Inf. Comput.},
  volume       = {278},
  pages        = {104597},
  year         = {2021},
  url          = {https://doi.org/10.1016/j.ic.2020.104597},
  doi          = {10.1016/J.IC.2020.104597},
  bibsource    = {dblp computer science bibliography, https://dblp.org}
}

@book{Halpern2003-HALRAU,
	author = {Joseph Y. Halpern},
	editor = {},
	publisher = {MIT Press},
	title = {Reasoning About Uncertainty},
	year = {2003}
}

@article{CastanedaDRV24,
  author       = {Armando Casta{\~{n}}eda and
                  Hans van Ditmarsch and
                  David A. Rosenblueth and
                  Diego A. Vel{\'{a}}zquez},
  title        = {Pattern Models: {A} Dynamic Epistemic Logic For Distributed Systems},
  journal      = {Comput. J.},
  volume       = {67},
  number       = {7},
  pages        = {2421--2440},
  year         = {2024},
  url          = {https://doi.org/10.1093/comjnl/bxae016},
  doi          = {10.1093/COMJNL/BXAE016},
  bibsource    = {dblp computer science bibliography, https://dblp.org}
}

@article{FuggerNS21,
  author       = {Matthias F{\"{u}}gger and
                  Thomas Nowak and
                  Manfred Schwarz},
  title        = {Tight Bounds for Asymptotic and Approximate Consensus},
  journal      = {J. {ACM}},
  volume       = {68},
  number       = {6},
  pages        = {46:1--46:35},
  year         = {2021},
  url          = {https://doi.org/10.1145/3485242},
  doi          = {10.1145/3485242},
  bibsource    = {dblp computer science bibliography, https://dblp.org}
}

@inproceedings{Akkoyunlu75,
author = {Akkoyunlu, E. A. and Ekanadham, K. and Huber, R. V.},
title = {Some constraints and tradeoffs in the design of network communications},
year = {1975},
isbn = {9781450378635},
publisher = {Association for Computing Machinery},
address = {New York, NY, USA},
url = {https://doi.org/10.1145/800213.806523},
doi = {10.1145/800213.806523},
booktitle = {Proceedings of the Fifth ACM Symposium on Operating Systems Principles},
pages = {67--74},
numpages = {8},
location = {Austin, Texas, USA},
series = {SOSP '75}
}

@inproceedings{gray78,
author = {Gray, Jim},
title = {Notes on Data Base Operating Systems},
year = {1978},
isbn = {3540087559},
publisher = {Springer-Verlag},
address = {Berlin, Heidelberg},
booktitle = {Operating Systems, An Advanced Course},
pages = {393--481},
numpages = {89}
}

@article{VargheseL96,
  author       = {George Varghese and
                  Nancy A. Lynch},
  title        = {A Tradeoff Between Safety and Liveness for Randomized Coordinated
                  Attack},
  journal      = {Inf. Comput.},
  volume       = {128},
  number       = {1},
  pages        = {57--71},
  year         = {1996},
  url          = {https://doi.org/10.1006/inco.1996.0063},
  doi          = {10.1006/INCO.1996.0063},
  bibsource    = {dblp computer science bibliography, https://dblp.org}
}

@book{Lynch96,
  author       = {Nancy A. Lynch},
  title        = {Distributed Algorithms},
  publisher    = {Morgan Kaufmann},
  year         = {1996},
  isbn         = {1-55860-348-4},
  bibsource    = {dblp computer science bibliography, https://dblp.org}
}

@misc{aspnes2025notestheorydistributedsystems,
      title={Notes on Theory of Distributed Systems}, 
      author={James Aspnes},
      year={2025},
      eprint={2001.04235},
      howpublished={arXiv 2001.04235},
      primaryClass={cs.DC},
      url={https://arxiv.org/abs/2001.04235}, 
}

@article{ProbabDEL-Kooi-2003,
	author = {Kooi, Barteld P. },
	date = {2003/09/01},
	doi = {10.1023/A:1025050800836},
	id = {Kooi2003},
	isbn = {1572-9583},
	journal = {Journal of Logic, Language and Information},
	number = {4},
	pages = {381--408},
	title = {Probabilistic Dynamic Epistemic Logic},
	url = {https://doi.org/10.1023/A:1025050800836},
	volume = {12},
	year = {2003}
}

@article{vanBenthemGerbrandyKooi2009,
	author = {van Benthem, Johan and Gerbrandy, Jelle and Kooi, Barteld},
	date = {2009/10/01},
	doi = {10.1007/s11225-009-9209-y},
	id = {van Benthem2009},
	isbn = {1572-8730},
	journal = {Studia Logica},
	number = {1},
	pages = {67--96},
	title = {Dynamic Update with Probabilities},
	url = {https://doi.org/10.1007/s11225-009-9209-y},
	volume = {93},
	year = {2009}
}

@book{FHMV1995,
  author       = {Ronald Fagin and
                  Joseph Y. Halpern and
                  Yoram Moses and
                  Moshe Y. Vardi},
  title        = {Reasoning About Knowledge},
  publisher    = {{MIT} Press},
  year         = {1995},
  url          = {https://doi.org/10.7551/mitpress/5803.001.0001},
  doi          = {10.7551/MITPRESS/5803.001.0001},
  isbn         = {9780262562003},
  bibsource    = {dblp computer science bibliography, https://dblp.org}
}

@article{Lamport78,
  author       = {Leslie Lamport},
  title        = {Time, Clocks, and the Ordering of Events in a Distributed System},
  journal      = {Commun. {ACM}},
  volume       = {21},
  number       = {7},
  pages        = {558--565},
  year         = {1978},
  url          = {https://doi.org/10.1145/359545.359563},
  doi          = {10.1145/359545.359563},
  bibsource    = {dblp computer science bibliography, https://dblp.org}
}

@book{Hintikka1962-HINKAB-4,
	address = {Ithaca, NY, USA},
	author = {Kaarlo Jaakko Juhani Hintikka},
	editor = {},
	publisher = {Cornell University Press},
	title = {Knowledge and Belief: An Introduction to the Logic of the Two Notions},
	year = {1962}
}

@inproceedings{DBLP:conf/stoc/HerlihyS93,
  author       = {Maurice Herlihy and
                  Nir Shavit},
  editor       = {S. Rao Kosaraju and
                  David S. Johnson and
                  Alok Aggarwal},
  title        = {The asynchronous computability theorem for t-resilient tasks},
  booktitle    = {Proceedings of the Twenty-Fifth Annual {ACM} Symposium on Theory of
                  Computing, May 16-18, 1993, San Diego, CA, {USA}},
  pages        = {111--120},
  publisher    = {{ACM}},
  year         = {1993},
  url          = {https://doi.org/10.1145/167088.167125},
  doi          = {10.1145/167088.167125},
  bibsource    = {dblp computer science bibliography, https://dblp.org}
}

@article{DBLP:journals/ipl/AlpernS85,
  author       = {Bowen Alpern and
                  Fred B. Schneider},
  title        = {Defining Liveness},
  journal      = {Inf. Process. Lett.},
  volume       = {21},
  number       = {4},
  pages        = {181--185},
  year         = {1985},
  url          = {https://doi.org/10.1016/0020-0190(85)90056-0},
  doi          = {10.1016/0020-0190(85)90056-0},
  bibsource    = {dblp computer science bibliography, https://dblp.org}
}

@article{DBLP:journals/jacm/FaginHV92,
  author       = {Ronald Fagin and
                  Joseph Y. Halpern and
                  Moshe Y. Vardi},
  title        = {What Can Machines Know? On the Properties of Knowledge in Distributed
                  Systems},
  journal      = {J. {ACM}},
  volume       = {39},
  number       = {2},
  pages        = {328--376},
  year         = {1992},
  url          = {https://doi.org/10.1145/128749.150945},
  doi          = {10.1145/128749.150945},
  bibsource    = {dblp computer science bibliography, https://dblp.org}
}

@book{vanDitmarsch2007-VANDEL-6,
	address = {Dordrecht, Netherland},
	author = {Hans van Ditmarsch and Wiebe van der Hoek and Barteld Kooi},
	editor = {},
	publisher = {Springer},
	title = {Dynamic Epistemic Logic},
	year = {2007}
}

@inproceedings{DBLP:conf/wdag/FraigniaudPR25,
  author       = {Pierre Fraigniaud and
                  Boaz Patt{-}Shamir and
                  Sergio Rajsbaum},
  editor       = {Dariusz R. Kowalski},
  title        = {Coordination Through Stochastic Channels},
  booktitle    = {39th International Symposium on Distributed Computing, {DISC} 2025,
                  Berlin, Germany, October 27-31, 2025},
  series       = {LIPIcs},
  pages        = {32:1--32:19},
  publisher    = {Schloss Dagstuhl - Leibniz-Zentrum f{\"{u}}r Informatik},
  year         = {2025},
  url          = {https://doi.org/10.4230/LIPIcs.DISC.2025.32},
  doi          = {10.4230/LIPICS.DISC.2025.32},
  bibsource    = {dblp computer science bibliography, https://dblp.org}
}

@inproceedings{DBLP:conf/stacs/GoubaultLR22,
  author       = {{\'{E}}ric Goubault and
                  J{\'{e}}r{\'{e}}my Ledent and
                  Sergio Rajsbaum},
  editor       = {Petra Berenbrink and
                  Benjamin Monmege},
  title        = {A Simplicial Model for KB4{\_}n: Epistemic Logic with Agents That
                  May Die},
  booktitle    = {39th International Symposium on Theoretical Aspects of Computer Science,
                  {STACS} 2022, Marseille, France (Virtual Conference), March 15-18,
                  2022},
  series       = {LIPIcs},
  pages        = {33:1--33:20},
  publisher    = {Schloss Dagstuhl - Leibniz-Zentrum f{\"{u}}r Informatik},
  year         = {2022},
  url          = {https://doi.org/10.4230/LIPIcs.STACS.2022.33},
  doi          = {10.4230/LIPICS.STACS.2022.33},
  bibsource    = {dblp computer science bibliography, https://dblp.org}
}

@article{DBLP:journals/corr/abs-2002-08863,
  author       = {Hans van Ditmarsch and
                  Eric Goubault and
                  J{\'{e}}r{\'{e}}my Ledent and
                  Sergio Rajsbaum},
  title        = {Knowledge and simplicial complexes},
  journal      = {CoRR},
  volume       = {abs/2002.08863},
  year         = {2020},
  url          = {https://arxiv.org/abs/2002.08863},
  eprinttype   = {arXiv},
  eprint       = {2002.08863},
  bibsource    = {dblp computer science bibliography, https://dblp.org}
}

@article{DBLP:journals/jacm/FaginH94,
  author       = {Ronald Fagin and
                  Joseph Y. Halpern},
  title        = {Reasoning About Knowledge and Probability},
  journal      = {J. {ACM}},
  volume       = {41},
  number       = {2},
  pages        = {340--367},
  year         = {1994},
  url          = {https://doi.org/10.1145/174652.174658},
  doi          = {10.1145/174652.174658},
  bibsource    = {dblp computer science bibliography, https://dblp.org}
}

@article{DBLP:journals/jacm/FischerLP85,
  author       = {Michael J. Fischer and
                  Nancy A. Lynch and
                  Mike Paterson},
  title        = {Impossibility of Distributed Consensus with One Faulty Process},
  journal      = {J. {ACM}},
  volume       = {32},
  number       = {2},
  pages        = {374--382},
  year         = {1985},
  url          = {https://doi.org/10.1145/3149.214121},
  doi          = {10.1145/3149.214121},
  bibsource    = {dblp computer science bibliography, https://dblp.org}
}

\appendix

\section{The system \texorpdfstring{$\mathsf{S5}$}{S5}}
The system $\mathsf{S5}$ consists of the axioms
\begin{gather}
\text{all propositional tautologies}\tag{Taut}\label{ax:taut}\\
K_p(\phi \rightarrow \psi) \rightarrow(K_p\phi \rightarrow K_p\psi)\tag{K}\label{ax:K}\\
K_p\phi \rightarrow \phi \tag{T}\label{ax:T}\\
K_p \phi \rightarrow K_pK_p \phi \tag{4}\label{ax:4}\\
\langle K_p\rangle \phi \rightarrow K_p\langle K_p\rangle  \phi \tag{5}\label{ax:5}
\end{gather} 
 as well as the inference rules modus ponens MP and $K_p$-necessitation $K_p$-Nec:
 
\[
\begin{prooftree}
  \hypo{A}
  \hypo{A \to B}
  \infer2{B}
\end{prooftree}
\qquad(\text{MP})\qquad
\begin{prooftree}
  \hypo{A}
  \infer1{K_p A}
\end{prooftree}
\qquad(K_p\text{-Nec})
\]

\section{Information levels}
We repeat the definition of the information level from~\cite{Lynch96}.
At the end of round $r$ in execution $\alpha$, process $p$ is in state $s_p(\alpha,r)$ and first computes its ``local'' happens-before relation~\cite{Lamport78} after receiving the messages of that round. Intuitively, the relation describes which information the process has at the end of round $r$.

\begin{definition}[Local happens-before relation]\label{def:leq_p}
Let $\alpha = (I,G,x)$ be an execution. At the end of round $r$, process $p$ is in state
\[
(i_p, \bar{x}_{p\mid r}, \rec_p(G,r)),
\]
and computes the \emph{local happens-before relation}~$\leq_p$ on the domain $N\times [r]$. We define~$\leq_p$ as the smallest relation satisfying:
    \begin{enumerate}
        \item \textbf{Local progression:} if $r<r'$, then $(p,r) \leq_p (p,r')$;
        \item \textbf{Message causality:} if $p$ receives $m_{q,r}$, then $(q,r-1) \leq_p (p,r)$;
        \item \textbf{Transitivity:} if $(q,r)\leq_p (s,r')$ and $(s,r')\leq_p (l,r'')$, then $(q,r)\leq_p(l,r'')$.
    \end{enumerate}
\end{definition}

\begin{remark}
    In contrast to the Definition of $\leq_p$ in~\cite{Lynch96}, we have shifted round indices by 1.
\end{remark}

After computing the local happens-before relation, process $p$ performs the \emph{level algorithm} (Definition~\ref{def:level:alg}). Intuitively, the level algorithm captures a process’s knowledge about what other processes know, and iteratively, what they know about others’ knowledge. Notice that in a full-information protocol, process $p$ can recompute process $q$'s level at the end of round $r-1$ locally if it received a message from process $q$ at the end of round $r$.

\begin{definition}[Level algorithm~\cite{Lynch96}]\label{def:level:alg}
Using the happens-before relation~$\leq_p$ computed in the current round $r$ (Definition~\ref{def:leq_p}), process $p$ locally computes its \emph{level} $\ell(p,\alpha, r)$:
\begin{enumerate}
     \item \textbf{Initialization:} the level of every process is initialized to 0;
    \item \textbf{Base case:} if there is some $q\neq p$ such that $(q,1)\not\leq_p (p,r)$, then $\ell(p,\alpha, r) = 0$;
    \item \textbf{Induction:} if 
    $(q,1)\leq_p (p,r)$ for all $q\neq p$, then define for each $q$:
    \[
    f_q := \max\{ \ell(q,\alpha, r') \mid (q,r')\leq_p (p,r)\},
    \]
    which is non-empty due to $(q,1)\leq_p (p,r)$ for all $q\neq p$. Next, set: 
    \[
    \ell(p,\alpha, r) = \min\{ f_q \mid q \neq p\} +1.
    \]
\end{enumerate}
\end{definition}

\begin{remark}
    Since $G$ can be arbitrary, the level algorithm applies to any full-information protocol with perfect recall, regardless of the underlying network topology.
\end{remark}

Intuitively, if a process has level $n$, it knows that all other processes have at least level $n-1$, and that there exists at least one process with exactly level $n-1$. In particular, if a process has level $1$, it knows the inputs of all other processes.

\section{Proofs for Section~\ref{sec:information levels intro}}
This section presents the proofs of Theorem~\ref{thm:epistemic levels}, Lemma~\ref{lem:prob_difference}, and Theorem~\ref{thm:lb}.

\subsection*{Theorem~\ref{thm:epistemic levels}}
\begin{proof}

    We show both directions separately via induction on $n$, starting from left to right. That is, we show that
    \[
    \ell(p,\alpha, R) = n \text{ implies } \CI[\CA],\alpha \models K_pE^{n-1}\phi_\alpha \land \neg K_pE^n \phi_\alpha.
    \]
    In words, if $p$ has level $n$ at the end of execution $\alpha$, then its nested knowledge of the inputs (represented by $\phi_\alpha$) is $n-1$ deep.
    
    In what follows we assume that $\alpha := (I,G,x).$\\
    
    \noindent \textbf{Base case}.
    Let $\ell(p, \alpha,R) = 1$, we need to show that 
    \[
    \CI[\CA], (I,G,x) \models K_p \phi_\alpha\land \neg K_pE\phi_\alpha.
    \]
    We show both sides of the conjunction separately, starting with $K_p\phi_\alpha$.
    
    By assumption, process $p$ has received information that every other process's level is 0. By the definition of the operational model in Section~\ref{subsec:operational model}, process $p$'s local state $s(p,\alpha,R)$ contains all inputs. Since there is a one-to-one correspondence between the operational model and the worlds of the execution model $\CI[\CA]$ (see Lemma~\ref{lem:correspnence}), the valuation $V^\CA$ assigns the same propositional input variables to each world $(I',G',x')$ that is indistinguishable from, i.e., $(I,G,x)\sim^\CA_p (I',G',x')$. Therefore, it holds that $\CI[\CA], (I',G',x') \models \phi_\alpha$, and by the definition of truth (see Definition~\ref{def:relational_truth}), we obtain $\CI[\CA],(I,G,x) \models K_p \phi_\alpha$. 
    
    The argument, that this is process $p$'s maximal nested knowledge about the inputs, i.e. $\CI[\CA],\alpha \models \neg K_pE\phi_\alpha$ relies on the same arguments. 
    If process $p$ has exactly level 1, then there exits a process, say $q$, for which $p$ considers it possible that it only has level 0. In the execution model $\CI[\CA]$, this corresponds to $p$ not being able to distinguish $\alpha$ from an execution $\alpha'$ with $\CI[\CA], \alpha' \models \neg K_q\phi_\alpha$. By Definition~\ref{def:relational_truth}, it holds that $\CI[\CA], (I,G,x)\models \neg K_pE\phi_\alpha$.\\
    

    \noindent\textbf{Inductive step}.
    Let  $\ell(p,\alpha,R) = n+1$. We need to show that
    \[
    \CI[\CA], (I,G,x) \models K_pE^{n} \phi_\alpha\land \neg K_pE^{n+1}\phi_\alpha.
    \]
    Again, we show both sides of the conjunction separately, starting with $K_pE^{n}\phi_\alpha$.
    
    By assumption process $p$ has received information that every other process's level is~$n$. By applying the induction hypothesis, we find that 
    \begin{equation}\label{eq:lem:3.3:1}
    \CI[\CA], (I,G,x) \models K_pE_{-p}E^{n-1}\phi_\alpha,
    \end{equation}
    where $E_{-p}$ is short for \emph{everybody except $p$}. A standard formula valid in $\mathsf{S5}$ for arbitrary $\psi\in \CL$ and $k\in \mathbb{N}$ is\footnote{This formula can be derived in the same way for $\psi\in \CL_\SFP$.}
    \[
    K_pK_qE^k\psi \rightarrow K_pE^k\psi.
    \]
    By setting $\psi = \phi_\alpha$,
    we obtain that $ \CI[\CA], (I,G,x) \models K_pE^{n-1}\phi_\alpha$. Indeed, since for the non-trivial cases of the problem it holds that $n>1$, and thus $N\setminus \{p\}\neq \emptyset$, we have that  
    \[
    E_{-p} \phi_\alpha \equiv \bigwedge_{a\in N\setminus\{p\}}K_q\phi_\alpha
    \]
    implies that there exists another process $q$ with $\CI[\CA], (I,G,x) \models K_pK_qE^{n-1}\phi_\alpha$. By axiom 4 (transitivity), we obtain that $ \CI[\CA], (I,G,x) \models K_pK_pE^{n-1}\phi_\alpha$. Combined with Equation~\eqref{eq:lem:3.3:1}, this implies 
    \begin{align*}
    \CI[\CA], (I,G,x) \models K_pK_pE^{n-1}\phi_\alpha  \land K_pE_{-p}E^{n-1}\phi_\alpha \\ &\iff \\\CI[\CA], (I,G,x) \models K_p(K_pE^{n-1}\phi_\alpha \land E_{-p}E^{n-1}\phi_\alpha)\\
    & \iff\\ \CI[\CA], (I,G,x) \models K_p(K_pE^{n-1}\phi_\alpha \land \bigwedge_{q\in N\setminus\{p\}}K_qE^{n-1}\phi_\alpha)\\
    &\iff \\\CI[\CA], (I,G,x) \models K_pE^{n}\phi_\alpha,
    \end{align*}
    as desired.\\ 
    
    The argument for showing that this is the maximal nested knowledge is the same as for the base case.\\
    
    We proceed by showing the claim from right to left. That is, we show that
    \[
    \CI[\CA],\alpha \models K_pE^{n-1}\phi_\alpha \land \neg K_pE^n \phi_\alpha  \text{ implies } \ell(p,\alpha, R) = n.
    \]

    \noindent\textbf{Base case}. Assume $\CI[\CA],(I,G,x) \models K_p \phi_\alpha \land \neg K_pE\phi_\alpha$, we need to show that $\ell(p,\alpha,R)=1$.
    By Definition~\ref{def:relational_truth}, it holds that $(I,G,x)\sim^\CA_p (I',G',x')$ implies $\CI[\CA],(I',G',x')  \models \phi_\alpha$. By the one-to-one correspondence in Lemma~\ref{lem:correspnence}, it must hold that that for all executions $\alpha'$ such that $s(p,\alpha',R) = s(p,\alpha, R)$, we have $I' = I$. However, this can only occur in the relational model if and only if process $p$ has heard that all other process have at least level 0. By similar reasoning as before, $\neg K_pE^n\phi_\alpha$ ensures that $p$'s level is exactly $1$.\\

    \noindent\textbf{Inductive step.}
    For the inductive step, we prove the contraposition, i.e.,
    \[
    \ell(p,\alpha, R) \neq 1 \text{ implies }\CI[\CA],(I,G,x)\models \neg K_pE^{n-1}\phi_\alpha \lor K_pE^n \phi_\alpha.
    \]
    
    Assume that $\ell(p,\alpha,R)<n$. The other case follows immediately because $E^{k+1}\psi \rightarrow E^k\psi$ is a standard consequence of axiom T ($K_p\psi \rightarrow \psi$). In particular, we will show that $\CI[\CA],(I,G,x) \models \neg K_pE^{n-1}\phi_\alpha$. Since $\ell(p,\alpha,R) < n$ , there exists a process $q$ from which $p$ has heard of at most level $n-2$. Let $t$ be the round in which $p$ has heard of $q$ last, i.e., $t$ is the maximal round for which $(q,t)\leq_p(p,R)$. Since there are finitely many rounds, the round $t$ exists.
    Let $G'$
    be the sequence of communication graphs in which all messages to $q$ after $t$ are dropped. Let $\alpha' = (I,G',x)$. By construction, it holds that $\mathsf{in}_p(G) = \mathsf{in}_p(G')$, and thus $(I,G,x) \sim_p (I,G',x)$. Moreover, since $\ell(q,\alpha',R) < n-1$ is true by assumption, the inductive hypothesis applies, by which it must hold that $\CI[\CA], (I,G',x) \models \neg K_qE^{n-2}\phi$, from which we conclude that $p$ must consider it possible in $(I,G,x)$ that $q$'s level is $n-2$, i.e., $\CI[\CA], (I,G,x) \models \neg K_pE^{n-1}\phi_\alpha$.\qedhere
    
\end{proof}

\subsection*{Lemma~\ref{lem:prob_difference}}
\begin{proof}

We define the events
\begin{itemize}
    \item $A := \{(I,G,\cdot) \in \Omega \mid \delta(s_i((I,G,\cdot), R)) = 1)\}$;
    \item $B := \{(I,G,\cdot)\in \Omega \mid\delta(s_j((I,G,\cdot), R)) = 1)\}$;
    \item $D := \{(I,G,\cdot) \in \Omega \mid \delta(s_i((I,G,\cdot), R) ) \neq \delta(s_j((I,G,\cdot), R))\}$.
\end{itemize}
Then
\begin{itemize}
    \item $\mathsf{Pr}_{I,G}[i \text{ decides }1 ] = \mathsf{Pr}_{I,G}[A]$;
    \item $\mathsf{Pr}_{I,G}[j \text{ decides }1] = \mathsf{Pr}_{I,G}[B]$;
    \item $\mathsf{Pr}_{I,G}[i \text{ and } j \text{ disagree }] = \mathsf{Pr}_{I,G}[D]$.
\end{itemize}
Observe that 
\[
A \Delta B \subseteq D,
\]
where $A \Delta B$ is the \emph{symmetric difference} $(A\setminus B) \cup (B \setminus A)$. A standard property of probability measures is:
\[
|\mathsf{Pr}_{I,G}[A]-\mathsf{Pr}_{I,G}[B]|\leq  \mathsf{Pr}_{I,G}[A\Delta B].
\]
By monotonicity ($A\Delta B \subseteq D$) and our assumption, we obtain that $\mathsf{Pr}_{I,G}[A\Delta B] \leq \prob_{I,G}[D]$ and $\prob_{I,G}[D] \leq \epsilon $, which yields the claim.


\end{proof}

\subsection*{Theorem~\ref{thm:lb}}
\begin{proof}
    We construct an explicit $pq$-chain. The worlds of our interest are associated to three different types of strategies (sequences of communication graphs) that the adversary can choose:
\begin{enumerate}
    \item $B_i$ with $0\leq i \leq R$ represents the strategy in which all messages up to round $i$ are received by everyone, and all later messages are dropped. 
    \item $A_i$ with $0\leq i \leq R-1$  represents the strategy in which only process $p$ does not receive any messages from round $i$ on. Before round $i$ process $p$ received all messages. 
    \item $C_i$ with $0\leq i \leq R-1$ represents the strategy  in which only process $q$ does not receive any messages from round $i$ on. Before round $i$ process $q$ received all messages. 
\end{enumerate} 
Since inputs are fixed to be 1 throughout our construction, we identify worlds of $\CI[\CA]$ by those sets only. Consider now the below excerpt of the execution model $\CI[\CA]$. Reflexive and transitive arrows are omitted for simplicity. It holds that 
\begin{equation}\label{eq:lb:1}
    B_i \sim^\CA_p A_i \text{ and } B_i \sim^\CA_q C_i,\;  \text{for all $0\leq i \leq R$},
\end{equation}
as well as
\begin{equation}\label{eq:lb:2}
    B_j \sim^\CA_q A_{j-1} \text{ and } B_j \sim^\CA_p C_{j-1},\; \text{ for all $1\leq j \leq R$}.
\end{equation}

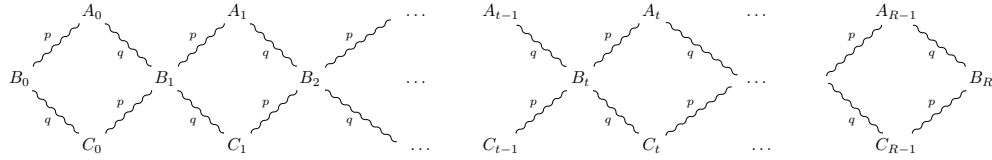
 \newcommand{\ara}{\ar@{~}[ur]^p}
 \newcommand{\arb}{\ar@{~}[dr]_q}
 \scalebox{0.65}
 {
\xymatrix{
&A_0\arb&&A_1\arb&&\;\;\;\;\;\;\dots&
A_{t-1}
\arb&&A_t\arb&\;\;\;\;\;\;\dots&&
A_{R-1}\arb
\\
B_0\ara\arb&&B_1\ara\arb&&B_2\ara\arb&\;\;\;\;\;\;\dots&
&B_t\ara\arb&&\;\;\;\;\;\;\dots&\ara\arb&&
B_R
\\
&C_0\ara&&C_1\ara&&\;\;\;\;\;\;\;\;\dots&
C_{t-1}\ara& 
&C_t\ara&\;\;\;\;\;\;\;\;\dots&
& C_{R-1}\ara
}
}

A $pq$-chain from $A_0$, a global state in which process $p$ decides the value $1$ with probability $0$, to $A_{R-1}$, a global state in which process $p$ decides the value $1$ with probability $1$, is given by:
\[
A_0 \sim^\CA_qB_1 \sim^\CA_p A_1 \sim^\CA_q B_2 \sim^\CA_p \cdots \sim^\CA_pA_{R-1}.
\]
Since $p$ ($q$) decides 1 in $A_0$ ($B_0$) with probability 0, the probability of both individually deciding 1 is at most $\epsilon$ in $B_1$. This naturally generalizes to $A_i,C_i$, and $B_{i+1}$. That is, for each segment $A_i \sim^\CA_q B_{i+1}\sim^\CA_p A_{i+1}$, the probability of $p$ deciding $1$ is at most $(i+1)\epsilon$. Thus, when in $A_{R-1}$, the probability is at most $R\epsilon$. By the task specification, the probability of $p$ deciding 1 in $A_{R-1}$ is 1, and therefore, it must hold that $R\epsilon \geq 1$. Rearranging the terms yields $\epsilon \geq R^{-1}$. Therefore $\CM^O,w\models\SFP_{\geq R^{-1}}\psi_{\mathsf{Dis}}$ for all worlds of the output model.
\end{proof}

\subsection*{Theorem~\ref{thm:lb_iis}}
\begin{proof}
    Throughout this proof, all processes have input 1.
    We employ the same technique as in the proof of Theorem~\ref{thm:lb}. Our goal is to find a $pq$-chain from a world in which a process $p$ decides the value 1 with probability 0 to another world in which it decides 1 with probability~1.  The probability of a process $p$ deciding 1 along the chain can be bounded by using similar arguments as in Theorem~\ref{thm:lb}, which yields a lower bound $\epsilon \geq \frac{1}{L}$, where $L$ is the length of the $pq$-chain of interest.

    It is a standard result that the worlds in $\mathcal{IIS}$ increase exponentially by a factor of 3 in each round~\cite{CastanedaDRV24}. Hence, when fixing two processes $p$ and $q$, the updated model for round $R$ naturally contains a $pq$-chain with $3^R$ many worlds. By fixing the inputs of the other processes and assuming that they receive all messages, we can concentrate on this \emph{line} (see Example~\ref{ex:msg_upd}).

    The worlds at the endpoints of that line (see Figure~\ref{fig:ex:msg_upd}) are the worlds where process $p$, respectively $q$, received no messages throughout all rounds. The world in the middle $w^*$ (i.e. at position $\frac{n+1}{2}$) is the world where both processes receive every message in all rounds.

    As in the proof of Theorem~\ref{thm:lb}, we are interested in the first world in which a process decides 1 with probability 1. In Theorem~\ref{thm:lb}, that world was $A_{R-1}$. Here, it is the world where the process $p$ (or $q$) received all messages and process $q$ (or $p$) received all except the last message. Since the updated model for $R$ rounds is a line, the path from the leftmost world to that world is of length $\Theta(3^{R})$. Notice that it is difficult to determine the exact length, as it changes with the parity of $R$.

    Along each transition, the probability of process $p$ deciding 1 can increase by at most $2\epsilon$ by Lemma~\ref{lem:prob_difference}. Indeed, if we have the segment
    \[
    w_1 \sim^\CA_q w_2,
    \]
    then:
    \begin{enumerate}
        \item $p$'s probability of deciding 1 in $w_2$ can be at most $\epsilon$ larger than $q$'s probability in $w_1$; and 
        \item $p$'s probability of deciding 1 in $w_1$ can be at most $\epsilon$ less than $q$'s probability in $w_1$.
    \end{enumerate}
     Thus, $p$'s probability can increase by at most $2\epsilon$ along $q$ transitions, and stays the same along $p$-transitions by Lemma~\ref{lem:indist_preserves_prob}. Because there are also $\Theta(3^{R})$ many $q$-transitions,
    the lower bound is therefore given, for some  constant  $c\in \mathbb{R}$:
    \[
    \epsilon \geq c3^{-R},
    \]
    by the same reasoning as in Theorem~\ref{thm:lb}.

\end{proof}

\end{document}